\documentclass[journal]{IEEEtran}

\usepackage{cite}

\usepackage[pdftex]{graphicx}
\graphicspath{{../Drawings/}{./Drawings/}}
\DeclareGraphicsExtensions{.pdf}

\usepackage{xcolor} 

\usepackage{amsmath,amssymb}

\usepackage{algorithm,algpseudocode}
\algdef{SE}[DOWHILE]{Do}{doWhile}{\algorithmicdo}[1]{\algorithmicwhile\ #1}%

\newtheorem{theorem}{Theorem}

\newtheorem{proposition}{Proposition}
\newtheorem{lemma}{Lemma}

\newtheorem{proof}{Proof}

\usepackage[caption=false,font=footnotesize]{subfig}

\usepackage{url}

\begin{document}

\title{Delta-Ramp Encoder for Amplitude Sampling \\and its Interpretation as Time Encoding}

\author{Pablo~Mart\'inez-Nuevo, \IEEEmembership{Member,~IEEE,}
        Hsin-Yu Lai,~\IEEEmembership{Student Member,~IEEE,}
        \\and~Alan~V.~Oppenheim,~\IEEEmembership{Life Fellow,~IEEE}
\thanks{\copyright~2019 IEEE. Personal use of this material is permitted. Permission from IEEE must be obtained for all other uses, in any current or future media, including reprinting/republishing this material for advertising or promotional purposes, creating new collective works, for resale or redistribution to servers or lists, or reuse of any copyrighted component of this work in other works.

This work was supported in part by Texas Instruments Leadership University Program. The work of P. Mart\'inez-Nuevo was also supported in part by Fundaci\'on Rafael del Pino, Madrid, Spain. The work of H. Lai was also supported in part by the Jacobs Fellowship and the Siebel Fellowship.

P. Mart\'inez-Nuevo was with the Department of Electrical Engineering and Computer Science, Massachusetts Institute of Technology. He is now with the research department at Bang \& Olufsen, 7600 Struer, Denmark (e-mail: pmnuevo@alum.mit.edu).

H. Lai and A. V. Oppenheim are with the Department of Electrical Engineering and Computer Science, Massachusetts Institute of Technology, Cambridge MA 02139 USA (e-mail: hsinyul@mit.edu; avo@mit.edu).}
\thanks{Digital Object Identifier 10.1109/TSP.2019.2904027}}

\markboth{IEEE Transactions on Signal Processing}%
{Shell \MakeLowercase{\textit{et al.}}: Bare Demo of IEEEtran.cls for IEEE Journals}

\maketitle
\begin{abstract}
The theoretical basis for conventional acquisition of bandlimited signals typically relies on uniform time sampling and assumes infinite-precision amplitude values. In this paper, we explore signal representation and recovery based on uniform amplitude sampling with assumed infinite precision timing information. The approach is based on the delta-ramp encoder which consists of applying a one-level level-crossing detector to the result of adding an appropriate sawtooth-like waveform to the input signal. The output samples are the time instants of these level crossings, thus representing a time-encoded version of the input signal. For theoretical purposes, this system can be equivalently analyzed by reversibly transforming through ramp addition a nonmonotonic input signal into a monotonic one which is then uniformly sampled in amplitude. The monotonic function is then represented by the times at which the signal crosses a predefined and equally-spaced set of amplitude values. We refer to this technique as amplitude sampling. The time sequence generated can be interpreted alternatively as nonuniform time sampling of the original source signal. We derive duality and frequency-domain properties for the functions involved in the transformation. Iterative algorithms are proposed and implemented for recovery of the original source signal. As indicated in the simulations, the proposed iterative amplitude-sampling algorithm achieves a faster convergence rate than frame-based reconstruction for nonuniform sampling. The performance can also be improved by appropriate choice of the parameters while maintaining the same sampling density.

\end{abstract}


%
\IEEEpeerreviewmaketitle

\section{Introduction}
%
%
%
%
\IEEEPARstart{T}{he} theoretical foundation of conventional time sampling typically relies on the sampling theorem for bandlimited signals \cite{Whittaker:1915aa,Kotelnikov:1933aa,Shannon:1949aa}, which states that bandlimited signals can be perfectly represented by infinite-precision amplitude values taken at equally-spaced time instants appropriately separated. In this paper, we propose a signal representation based on equally-spaced amplitude samples with infinite-precision timing information. We introduce the delta-ramp encoder that generates a time encoded version of the input signal and show how this sampling and reconstruction process can be theoretically analyzed based on the amplitude sampling concept also introduced in this paper.

Signal representation based on discrete amplitudes and continuous time has previously been studied and utilized in a number of contexts. In \cite{Bond:1958aa} signal representation consists of the real and complex zeros of a bandlimited signal. Logan's theorem \cite{Logan:1977aa} characterizes a subclass of bandpass signals that can be completely represented, up to a scaling factor, by their zero crossings. Practical algorithms for recovery from zero crossings of periodic signals in this class have been proposed in \cite{Roweis:1998ac}. Arbitrary bandlimited signals can also be implicitly described by the zero crossings of a function resulting from an invertible transformation \cite{Haavik:1966aa,Bar-David:1974aa,Kumaresan:2000aa}---for example, the addition of a sinewave \cite[Theorem 1]{Duffin:1938aa}. In principle, interpolation is possible through Hadamard's factorization \cite[Chapter 5]{Stein:2003aa} although there are more efficient techniques in terms of convergence rate \cite{Selva:2012aa,Kay:1986aa,Sreenivas:1992aa,Kumaresan:2010aa}. Zero-crossings have also been studied in relation to wavelet transforms \cite{Mallat:1991aa}. In this case, stable reconstruction can be achieved by including additional information about the original signal.

The extension from zero crossings to multiple levels, in the context of data compression, was investigated in \cite{Mark:1981aa}. In that work, a sample is generated whenever the source signal crosses a predefined set of threshold levels. The time instants of the crossings and the level-crossings directions were utilized to represent the signal although time was still quantized due to practical considerations. A practical continuous-time version of level-crossing sampling was later proposed in \cite{Tsividis:2003aa}. Asynchronous delta modulation \cite{Inose:1966aa} is, also, in some sense, a precursor of level-crossing sampling since it generates a positive or negative pulse at time instants when the change in signal amplitude surpasses a fixed quantity. 

In the context of asynchronous sigma-delta modulation systems,

In the context of asynchronous sigma-delta modulation systems, the connection between time-based representation and local averages of bandlimited signals was shown in \cite{Lazar:2003aa} where frame-based reconstruction can be carried out \cite[Theorem 7]{Feichtinger:1994aa}. This sampling process can then be viewed as a representation of a signal as a stream of pulses where processing can be performed directly in the pulse domain \cite{McCormick:2012aa}.

In this paper, we study the time encoding process of the delta-ramp encoder (see Fig.~\ref{fig:AmpSamp_Implementation}) and reconstruction from the generated time sequence. We show that this system can be analyzed theoretically based on the more general concept of amplitude sampling with the signal represented by the time sequence of equally-spaced level crossings of a monotonic transformation of the input signal as would be generated, for example, by an amplitude quantizer with equal step sizes. In principle, if a signal were monotonic, then the crossings of equally-spaced amplitude levels would generate an ordered time sequence $\{t_n\}$ which could be considered as a representation of the signal. Under appropriate conditions, this corresponds to uniform sampling in amplitude with the signal information contained in the time sequence $\{t_n\}$. Nonmonotonic signals can be reversibly transformed into monotonic ones which are then uniformly sampled in amplitude. We refer to this technique as amplitude sampling.

As discussed in Section \ref{section:Delta-Ramp Encoder} where we introduce the delta-ramp encoder, when the reversible transformation consists of adding a ramp with appropriate slope, a practical implementation to generate the identical ordered time sequence $\{t_n\}$ is the delta-ramp encoder shown in Fig.~\ref{fig:AmpSamp_Implementation}. The time sequence generated by the delta-ramp encoder and that obtained by uniform amplitude sampling after ramp addition are identical. For the theoretical analysis of the delta-ramp encoder in this paper, we utilize the interpretation of the time sequence $\{t_n\}$ as derived from uniform amplitude sampling of the monotonic function obtained by ramp addition. Section \ref{section:amplitude sampling} defines the general concept of amplitude sampling. In sections \ref{section:RampAddition}, \ref{section:SpectralProperties}, and \ref{section:TimeDomainProperties} we derive duality as well as time- and frequency-domain properties relating the functions present in the transformation. The structure of these functions suggest an iterative reconstruction algorithm for numerical recovery of the source signal from the amplitude samples. This algorithm is discussed in Section \ref{section:ReconsAmpSamp} with simulations and comparisons with frame-based reconstruction from nonuniform time samples.

Throughout the paper, we refer to $\hat{f}$ as the Fourier transform of the function $f$ given by
\begin{equation}
\label{eq:FTx_pairs1}
\hat{f}(\xi)=\int_{\mathbb{R}}f(t)e^{-i2\pi \xi t}\mathrm{d}t,\ \xi\in\mathbb{R}.
\end{equation}
The Fourier inversion formula then takes the following form
\begin{equation}
\label{eq:FTx_pairs2}
f(t)=\int_{\mathbb{R}}\hat{f}(\xi)e^{+i2\pi \xi t}\mathrm{d}\xi,\ t\in\mathbb{R}.
\end{equation}
Note that the units for $\xi$ can be interpreted to be Hz. We say that a function $f$ is of moderate decrease or moderate decay if it is continuous and there exists $A>0$ such that $|f(t)|\leq A/(1+t^2)$ for all $t\in\mathbb{R}$.

\section{Delta-Ramp Encoder}
\label{section:Delta-Ramp Encoder}
The delta-ramp encoder is represented by the block diagram depicted in Fig.~\ref{fig:AmpSamp_Implementation}. The level detector produces an impulse at times at which the input signal reaches the value $\Delta$. For ease of illustration, assume the ramp-segment generator initiates a ramp with slope $\alpha>0$ that abruptly shifts down by $\Delta$ in amplitude whenever an impulse arrives. Assume $\alpha$ is chosen such that $\tilde{g}(t)$ is monotonic in each interval between successive impulses.

\begin{figure}[thpb]
\centering
\includegraphics[width=0.85\columnwidth]{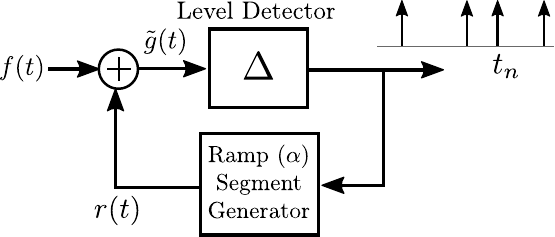}
\caption{Equivalent representation of the amplitude sampling process.}
\label{fig:AmpSamp_Implementation}
\end{figure}

Fig.~\ref{fig:AmpSamp_Implementation_explanation} shows an example of the signals involved in the process. By construction, the ramp segments of the function $r(t)$ present the same slope. This manifests itself in the presence of an continuous ramp of slope $\alpha$ separated by multiples of $\Delta$ for each corresponding segment. Consequently, the function $\tilde{g}(t)$ satisfies the following 
\begin{equation}
\tilde{g}(t)=f(t)+\alpha t-k\Delta
\end{equation}
for $t\in(t_k,t_{k+1}]$ and $k\in\mathbb{Z}$. Thus, 
\begin{equation}
\tilde{g}(t_{k+1})=\Delta=f(t_{k+1})+\alpha t_{k+1}-k\Delta
\end{equation}
which gives $(k+1)\Delta=g(t_{k+1})$ for all $k\in\mathbb{Z}$ where $g(t)=\alpha t+f(t)$. Consider now the time instants $\{t_n\}$ that satisfy $g(t_n)=n\Delta=\alpha t_n+f(t_n)$. As a result of the one-to-one correspondence between amplitude values and time instants due to the monotonicity of $g(t)$, it follows that $\{t_k\}=\{t_n\}$. Thus, the delta encoder generates impulses at the same time instants at which $g(t)$ crosses the set of amplitude levels $\{n\Delta\}$.

In summary, the delta-ramp encoder produces a representation of the input signal as a sequence of time instants, or time codes. This time encoding mechanism can be alternatively viewed as level-crossing sampling of the function $g(t)$ or nonuniform sampling of $f(t)$, i.e. $f(t_n)=n\Delta-\alpha t_n$. Moreover, the function $g(t)$, assuming appropriate regularity conditions, has an inverse function $t(g)$ which is effectively sampled uniformly in the amplitude domain with samples corresponding to these time instants. Therefore, the sampling process of the delta-ramp encoder can be interpreted as uniformly sampling the function $t(g)$. In principle, it is possible to generalize this concept by considering any transformation that generates a monotonic function $g(t)$. We formalize this concept in the next section.

\begin{figure}[thpb]
\centering
\includegraphics[width=0.6\columnwidth]{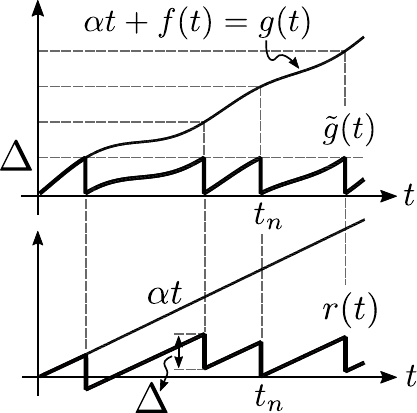}
\caption{Illustration of the different waveforms involved in the system shown in Fig.~\ref{fig:AmpSamp_Implementation}.}
\label{fig:AmpSamp_Implementation_explanation}
\end{figure}

\section{Principle of Amplitude Sampling}
\label{section:amplitude sampling}
Amplitude sampling and reconstruction as developed in this paper is then based on the principle of reversibly representing and then sampling a time function $g(t)$ in the form $t(g)$ and then sampling in $g$. This requires that $g(t)$ be monotonic which means that if the source signal is nonmonotonic, it must first be reversibly transformed into a strictly monotonic function through a transformation $\phi$. As illustrated in Fig.~\ref{fig:AmpSampPrinciple}, the resulting function $\phi(f(t))$ is then uniformly sampled. The time instants $\{t_n\}$ at which $\phi(f(t))$ crosses the predefined set of amplitude values $\{n\Delta\}$ implicitly represent the source signal, i.e. $\phi(f(t_n))=n\Delta$ where $\Delta>0$ is the separation between consecutive levels. Each of the time instants is paired exactly with one amplitude level. Thus, there exists a one-to-one correspondence between amplitude values and time instants. The sequence of time instants together with knowledge of $\Delta$ is sufficient information to describe the sampling process. Thus, it can be interpreted as a form of time encoding.

Amplitude sampling corresponds to signal-dependent nonuniform time sampling with the sampling density dependent on the source signal and the choice of the transformation $\phi$. 

\begin{figure}[thpb]
\centering
\includegraphics[scale=1.1]{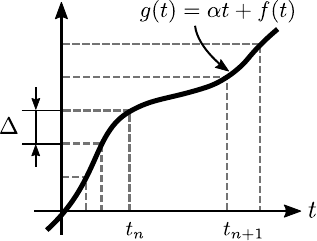}
\caption{Principle of amplitude sampling based on a transformation $\phi$ of the source signal $f$ resulting in a monotonic function $\phi(f(t))$.}
\label{fig:AmpSampPrinciple}
\end{figure}

We have shown in Section \ref{section:Delta-Ramp Encoder} that, when a ramp of appropriate slope is used, amplitude sampling is equivalent to the samples generated by a delta-ramp encoder. In fact, it can be shown that many delta-modulation systems can be interpreted as amplitude sampling. For a detailed analysis of the latter, the interested reader is referred to \cite[Chapter 4]{Martinez-Nuevo:2016aa}.

\section{Transformation by Ramp Addition}
\label{section:RampAddition}
There exist a myriad of transformations $\phi$ that can potentially generate a monotonic function from a given $f$. Among the simplest is  the addition of a ramp with a sufficiently large slope. Suppose the original signal $f$ is continuous, and it is possible to construct the strictly monotonic function $g(t)=\alpha t+f(t)$ for some $\alpha\in\mathbb{R}$. Then, the sampling process consists of the sequence of time instants $\{t_n\}$ satisfying $g(t_n)=\alpha t_n+f(t_n)=n\Delta$ for some $\Delta>0$. 

%
As indicated earlier, for analysis purposes in this paper, it is convenient to interpret the time sequence $\{t_n\}$ as resulting from sampling uniformly in amplitude the monotonic function $u=g(t)=\alpha t+f(t)$. In the context of this transformation, there exists an inverse function $g^{-1}(u)$ that we choose to express in the form $g^{-1}(u)=u/\alpha+h(u)$ for some amplitude-time function $h$. This interpretation suggests that this transformation can also be viewed as a mapping from $f$ to the associated function $h$.

\subsection{Mapping between $f$ and $h$}
The addition of a ramp represents a mapping, parametrized by the slope of the ramp, between the original signal and the function $h$. We denote this mapping by $M_\alpha$, i.e. $M_\alpha f=h$ which can be viewed as the addition of the ramp to obtain the monotonic function $g$ and, after inverting $g$, subtracting the ramp $u/\alpha$ to obtain $h$.  The reverse procedure to recover $f$ from $h$ consists of adding a ramp of slope $u/\alpha$ to $h$ and utilizing the invertibility of $g^{-1}$ as well as the correspondence between $g$ and $f$. This inverse mapping is denoted by  $M_{\alpha^{-1}}$ and satisfies $M_{\alpha^{-1}}h=f$.  Fig.~\ref{fig:Tx_ftohtof} illustrates the one-to-one correspondence between $f$ and $h$. These mappings are also summarized in equation form as \cite{Lai:2016ab}

\begin{equation}
\label{eq:JanesEq}
\begin{split}
f(t)=&-\alpha h(f(t)+\alpha t),\\
h(u)=&-\frac{1}{\alpha}f(h(u)+\frac{u}{\alpha}).
\end{split}
\end{equation}

\begin{figure}[thpb]
\centering
\includegraphics[scale=0.8]{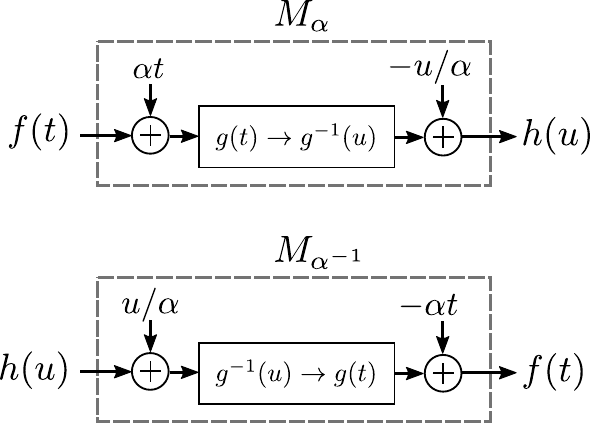}
\vspace{4mm}
\caption{Illustration of the invertibility of the transformation between $f$ and $h$ when $g(t)=\alpha t+f(t)$ and $g^{-1}(u)=u/\alpha+h(u)$.}
\label{fig:Tx_ftohtof}
\end{figure}

As is evident from Fig.~\ref{fig:Tx_ftohtof} and (\ref{eq:JanesEq}) there is a duality between $M_\alpha$ and its inverse. It is possible to interpret (\ref{eq:JanesEq}) as a signal-dependent warping operation that obtains $f$ from $h$ and vice versa. The addition of a ramp in amplitude sampling also generates an underlying mapping, dependent on $f$ or $h$, between time $t$ and amplitude $u$. Both mappings can be easily seen from (\ref{eq:JanesEq}) in its matrix form and the corresponding inverse matrix:
\begin{equation}
\label{eq:behaviorJanesEq}
\left(\begin{array}{c}f(t) \\t\end{array}\right)=\left(\begin{array}{cc}-\alpha & 0 \\1 & 1/\alpha\end{array}\right)\left(\begin{array}{c}h(u) \\u\end{array}\right),
\end{equation}

\begin{equation}
\label{eq:Behavior_h}
\left(\begin{array}{c}h(u) \\u\end{array}\right)=\left(\begin{array}{cc}-1/\alpha & 0 \\1 & \alpha\end{array}\right)\left(\begin{array}{c}f(t) \\t\end{array}\right).
\end{equation}
The duality implies that any properties of $h$ inherited by assumptions made on $f$ hold for $f$ if the same assumptions are instead imposed on $h$.

\subsection{The Sampling Process}
Amplitude sampling produces a sequence of time instants corresponding to $n\Delta=g(t_n)=\alpha t_n+f(t_n)$ where $\Delta>0$. This sampling process results then in $g^{-1}$ and $h$ being uniformly sampled in amplitude, i.e. 
\begin{equation}
g^{-1}(n\Delta)=n\Delta/\alpha+h(n\Delta)
\end{equation}
 and \begin{equation}
 h(n\Delta)=t_n-n\Delta/\alpha.
\end{equation}

\subsection{Sampling Density}
As noted earlier,  amplitude sampling in the form presented here can be viewed as equivalent to nonuniform time sampling. In this setting, stable reconstruction algorithms typically impose conditions on the sequence of sampling instants as for example the Landau rate \cite{Landau:1967aa} for bandlimited signals. In order to gain insight into the time-sampling density inherent in our amplitude-sampling process, assume the source signal $f$ has a bounded derivative, i.e. $|f'(t)|\leq B$ for some $B>0$ and that  $|\alpha|>B$ so that the function $g(t)=\alpha t+f(t)$ is strictly monotonic. Then  the time between successive samples satisfies the inequality  
\begin{equation}
\label{eq:boundssamplingset}
\frac{\Delta}{|\alpha|+B}\leq|t_{n+1}-t_n|\leq\frac{\Delta}{|\alpha|-B}.
\end{equation}
where $\Delta>0$ is the separation between consecutive amplitude levels.
 
The bounds in (\ref{eq:boundssamplingset}) are consistent with intuition. For example, assume that $\alpha$ is positive. The derivative of $g$ is bounded by $\alpha+B$ which provides the minimum attainable time separation between crossings. Similarly, the maximum separation is essentially limited by $\alpha-B$. The quantization step $\Delta$ represents the change in amplitude necessary to produce a sample. Additionally, when $\alpha$ achieves sufficiently large values, the bounds for time separation become closer, or equivalently, the time sequence becomes more uniform. We can observe this effect in (\ref{eq:boundssamplingset}) where amplitude is approximately a scaled version of the time axis.

\subsection{Iterative Algorithm for the Realization of $M_\alpha$}
In this section, we propose an iterative algorithm for the implementation of $M_\alpha$ to generate $h(u)$ from $f(t)$. By duality, an equivalent algorithm can be used for the implemention of  $M_{\alpha^{-1}}$ to generate $f(t)$ from $h(u)$. For ease of illustration, we consider a modified version of the transformation $M_\alpha$ by considering $\tilde{h}(u)=-\alpha h(\alpha u)$. This transformation, which we denote by $\tilde{M}_\alpha$, is then given by:
\begin{equation}
\label{eq:JanesEq_modified}
\begin{split}
f(t)&=\tilde{h}(\frac{1}{\alpha}f(t)+t),\\
\tilde{h}(u)&=f(-\frac{1}{\alpha}\tilde{h}(u)+u).
\end{split}
\end{equation}

Equations (\ref{eq:JanesEq_modified}) form the basis for the iterative algorithm formalized in the following theorem. 

\begin{theorem}
\label{thm:IterationTheorem}
Let the function $f$ be Lipschitz continuous with constant $<\alpha$ and suppose that $\sup_{t\in\mathbb{R}}|f(t)|\leq A$. Then, the function values $\tilde{h}(u)=(\tilde{M}_\alpha f)(u)$ for $u\in\mathbb{R}$ can be obtained by the iteration
\begin{equation}
\tilde{h}_{n+1}(u)=f(u-\frac{1}{\alpha}\tilde{h}_n(u))
\end{equation}
for $n\geq0$ where $\tilde{h}_0(u)=f(u)$ and $\tilde{h}_n(u)\to \tilde{h}(u)$ as $n\to\infty$.
\end{theorem}

The detailed proof is carried out in Appendix \ref{app:iterativeMalpha}.

\begin{figure}[thpb]
\centering
\includegraphics[width=0.85\linewidth]{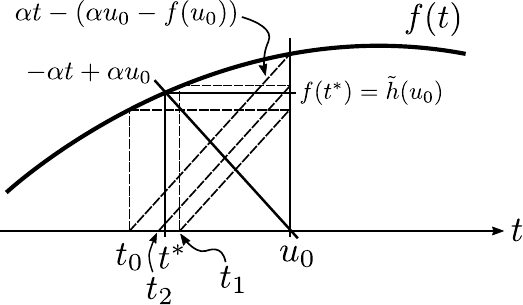}
\caption{Illustration of the iteration described in Theorem \ref{thm:IterationTheorem} with the initialization $t_0=u_0$.}
\label{fig:IterationM_a}
\end{figure}

As used in the preceding, the value of $\tilde{h}(u_0)$ can be obtained from the first equality in  (\ref{eq:JanesEq_modified}). In particular, $\tilde{h}(u_0)=f(t^*)$ where $t^*$ is the value that satisfies $u_0=t^*+f(t^*)/\alpha$. The solution is unique since the slope of the ramp, in absolute value, is always greater than the maximum value of the derivative of $f$. As shown in Fig.~\ref{fig:IterationM_a}, $t_0$ is the time instant at which the ramp $\alpha t_0-\alpha t$ intersects the function $f(t)$. In the same way, the value of the $(n+1)$-th iteration can be viewed as the solution of $\alpha t-(\alpha u_0-f(t_n))=0$. In other words, we iteratively construct a straight line passing through the point $(u_0,f(t_n))$. The intersection with the horizontal axis then corresponds to the value of $t_{n+1}$.

Note that the process for recovering $h$ from $f$ is analogous to the one presented in Theorem~\ref{thm:IterationTheorem}, i.e. the iteration takes the form $f_{n+1}(t)=\tilde{h}(f_n(t)/\alpha+t)$.

\section{Spectral Properties}
\label{section:SpectralProperties}{}
Assumptions made on the source signal $f$ are naturally reflected in the structure of $h$. In this section, we assume that $f$ is a bandlimited function and derive properties regarding the spectral content of the amplitude-time function $h$. The duality between $f$ and $h=M_\alpha f$ implies that similar conclusions can be made about $f$ when $h$ is assumed to be bandlimited.

In exploring the spectral content of $h$ we assume that $f$ is bandlimited to $\sigma$ rad/s  with $\sigma>0$ and bounded in amplitude, i.e. $|f(t)|<A$ for some $A$.  We further assume that the decay of $f(t)$ for $t$ real satisfies $|f(t)|\leq A/(1+t^2)$. In principle, the extension to square-integrable functions is straightforward. With our assumptions on $f$, Bernstein's inequality \cite{Bernstein:1926aa} provides the bound $|f'(t)|\leq A\sigma$ for all $t\in\mathbb{R}$. This bound gurantees that the function $u$ defined as 
\begin{equation}
u=g(t)=\alpha t+f(t).
\end{equation}
will be strictly monotonic whenever $|\alpha|>A\sigma$. The function $h$ is then given by $h(u)=g^{-1}(u)-u/\alpha$. From  Theorem \ref{thm:h_expdecay} below it follows that the decay of the Fourier transform of $h$, denoted by $\hat{h}(\xi)$, satisfies $\hat{h}(\xi)=\mathcal{O}(e^{-2\pi|\xi|b})$ as $\xi\to\infty$ where $b>0$ is determined by the difference $|\alpha|-A\sigma$.

\begin{theorem}
\label{thm:h_expdecay}
Let $f(t):\mathbb{R}\to\mathbb{R}$ be a continuous function bandlimited to $\sigma>0$ rad/s. Assume further that $|f(t)|\leq A/(1+t^2)$ for all $t\in\mathbb{R}$ and some $A>0$. Construct the function
\begin{equation} 
u=g(t)=\alpha t+f(t)
\end{equation} 
for $|\alpha|>A\sigma$. Then, there exists $g^{-1}(u)$ for all $u\in\mathbb{R}$ and a constant $C>0$ such that the Fourier transform of $h(u)=g^{-1}(t)-u/\alpha$ satisfies $|\hat{h}(\xi)|\leq Ce^{-2\pi |\xi|b}$ for any $0\leq b<a$ such that
\begin{equation}
\label{eq:h_expdecay}
a=\frac{|\alpha|}{\sigma}\log\Big(\frac{|\alpha|}{A\sigma}\Big)-\frac{|\alpha|-A\sigma}{\sigma}.
\end{equation}
and $\xi\in\mathbb{R}$.
\end{theorem}

The detailed proof is carried out in Appendix \ref{app:h_expdecay}.

As anticipated, the rate of decay of the Fourier transform at infinity depends on $|\alpha|-A\sigma$. The difference is logarithmic in the first term and linear in the second one. The larger the difference the faster the decay at infinity. Note that $a>0$ always holds since $|\alpha|>A\sigma$. Assuming $\alpha>0$, this difference is precisely impacting the highest slope portions in $h$, or, equivalently, the regions in which $f'$ is smallest. The underlying reason being that the derivative of $g^{-1}$ is the reciprocal of $g$, i.e. $(g^{-1}(u))'=1/g'(g^{-1}(u))$ for all $u\in\mathbb{R}$. Informally, it is the tilted regions in the shape of $h$ are responsible, to some extent, for the high-frequency content.

It should be emphasized that any bandlimited function will naturally be in the class of signals whose spectrum exhibits at least exponential decay  at infinity.  However, Theorem \ref{thm:h_nonBL} stated below asserts that $f$ and $h$ cannot be simultaneously bandlimited. The precise statement in the description of the theorem guarantees this property with the possible exception of, at most, one value of $\alpha$. For practical purposes, we can ignore this isolated case.

\begin{theorem}
\label{thm:h_nonBL}
Under the conditions of Theorem \ref{thm:h_expdecay} and unless $f$ is constant, the function $h$ is nonbandlimited for every $\alpha>A\sigma$ with at most one exception.
\end{theorem}

The detailed proof is carried out in Appendix \ref{app:h_nonBL}.

In the singular case, in which $f$ is a constant,  it can be shown through the constructive process of $M_\alpha$ by ramp addition that $h$ is constant as well, specifically, for $f=A$, then $h=-A/\alpha$. From another point of view, according to (\ref{eq:JanesEq}), the function $f$ results, in general, from $h$ with a nonlinear warping of the independent variable. When either of the two functions is constant, the warping is affine. Therefore, in this case, the bandlimited property  is preserved \cite{Azizi:1999aa}. In our context the conclusion follows directly from (\ref{eq:JanesEq}) that if either $f$ or $h$ is constant, the other must be also. 
 
More generally, if $|\alpha|$ increases significantly, the warping function becomes approximately linear since $f(t)$ is negligible compared to $\alpha t$, i.e. $f(t)\approx -\alpha h(\alpha t)$.  This is consistent with (\ref{eq:h_expdecay}) where an increase of $|\alpha|$ produces a faster decay at infinity of $\hat{h}$.

\section{Time-Domain Decay Properties}
\label{section:TimeDomainProperties}
In some sense, $h$ inherits characteristics of $f$ since it is a  "time-warped" version of $f$. In this section, we show the connections between the properties of $f$ and $h$ in the time domain with the relationship between $f$ and $h$ as specified in (\ref{eq:JanesEq}) which explicitly requires that the slope of the ramp added to $f$ and the slope of the ramp subtracted to obtain $h$ be exact inverses. Intuitively, it is not surprising that the function $h$ should present decay properties similar to those of $f$ once the unbounded growth of the ramp component has been subtracted. In particular, when the slopes of the two ramps are reciprocals of each other, the decay of $h$ will match that of $f$. Otherwise $h$ does not decay appropriately on the real line (see Proposition \ref{prop:boundedness} in Appendix \ref{app:BoundednessSlopes}).

The transformation $M_\alpha$ also has an impact on the $L^p$ norms of the respective functions with the parameter $\alpha$ playing a crucial role. It can be shown---refer to Proposition \ref{prop:Lpnorms} in Appendix \ref{app:BoundednessSlopes}---that the decay on the real line of both functions is related by
\begin{equation}
||h||_p=\frac{1}{\alpha^{1-\frac{1}{p}}}||f||_p,\ p\in[1,\infty].
\end{equation}

Not only does $h$ belong to $L^p(\mathbb{R})$ if $f$ does, but their respective norms are also related by a scaling factor which is precisely $\alpha$. Indeed, for very large values of $|\alpha|$, the ramp approaches the vertical axis, thus reducing the range of $h$ and decreasing the norm.

In terms of a sense of distance, consider $h_1=M_\alpha f_1$ and $h_2=M_\alpha f_2$. The transformation $M_\alpha$ preserves the $L^1$ distance (see Proposition \ref{prop:L1norms} in Appendix \ref{app:BoundednessSlopes}), i.e. 
\begin{equation}
||h_1-h_2||_1=||f_1-f_2||_1.
\end{equation}
By duality, these properties hold irrespective of the role of each function as an input or output.

\section{Reconstruction from Delta-Ramp Encoding}
\label{section:ReconsAmpSamp}
The time encoding performed by the delta-ramp encoder can be seen, under the amplitude sampling perspective, as signal-dependent nonuniform time sampling of the source signal $f$ based on uniform time sampling of the associated amplitude-time function $h$. If $f$ is bandlimited, then as was shown in Section \ref{section:SpectralProperties} $h$ is not bandlimited and consequently cannot be exactly reconstructed through bandlimited interpolation. Our reconstruction approach begins by initially using sinc interpolation as an approximation. This is then extended to  an iterative algorithm that achieves accurate recovery. Throughout this entire section, we assume that the source signal $f$ is bandlimited to $\sigma$ rad/s, and $|f(t)|\leq A/(1+t^2)$ for $A>0$ and all $t\in\mathbb{R}$.

\subsection{Bandlimited Interpolation Algorithm (BIA)}
The approximate reconstruction of $f$ based on sinc interpolation of $h$ is depicted in Fig.~\ref{fig:BandlimitedApprox}. From this approximation to $h$ an approximation to $f$ is generated through $M_{\alpha^{-1}}$ which is then lowpass filtered since $f$ is assumed to be bandlimited. In particular, the D/C system is defined by the relationship
\begin{equation}
\label{eq:bandlimitedinterpolation}
h_{\Delta}(u)=\sum_{n\in\mathbb{Z}}h(n\Delta)\mathrm{sinc}(u/\Delta-n).
\end{equation}
Note that the samples of $h$ are related to $\{t_n\}$ as $h(n\Delta)=t_n-n\Delta/\alpha$, $n\in\mathbb{Z}$. The motivation to perform bandlimited interpolation from the samples of $h$ is based on the exponential decay of its spectrum. Since $h$ is nonbandlimited, the aliasing error can be characterized by the following bound (see Proposition \ref{prop:BLerror} in Appendix \ref{app:infinityerror})
\begin{equation}
\label{eq:Linftyerror}
||h-h_{\Delta}||_{\infty}\leq \frac{C'}{a}e^{-\pi\frac{b}{\Delta}}
\end{equation}
for any $0\leq b<a$ and some $C'>0$ where $a$ is given in (\ref{eq:h_expdecay}).

Then, the error in (\ref{eq:Linftyerror}) is then controlled both by the difference $|\alpha|-A\sigma$ and the quantization step size $\Delta$. As already discussed, increasing the difference $|\alpha|-A\sigma$ produces, in some sense, a function $h$  with a faster high-frequency spectral decay and therefore one that is more approximately bandlimited. Since the quantization step size $\Delta$ also determines the sampling density of $h$, by decreasing $\Delta$ the aliasing error of the bandlimited interpolation is also reduced. 

By an appropriate combination of a sufficiently large $\alpha$ and/or a sufficiently small $\Delta$, the function $h_\Delta(u)+u/\alpha$ can be assumed to be invertible. We then obtain $f_\Delta=M_{1/\alpha}h_{\Delta}$. The function $f_\Delta$ is nonbandlimited since $h_\Delta$ is bandlimited (see Theorem \ref{thm:h_nonBL}). Thus, we obtain the bandlimited interpolation $\tilde{f}$ by passing $f_\Delta$ through a lowpass filter with cutoff frequency $\sigma$ rad/s.

\begin{figure}[thpb]
\centering
\includegraphics[scale=0.48]{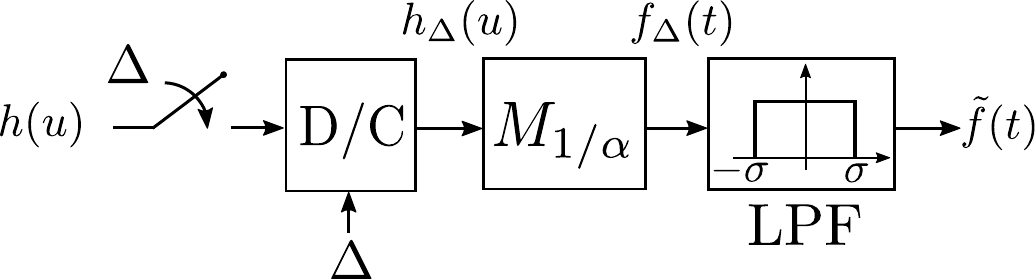}
\caption{Approximate reconstruction procedure for a bandlimited source signal $f$ such that $h=M_\alpha f$. The block D/C is a discrete-to-continuous operation involving sinc interpolation with period $\Delta$.}
\label{fig:BandlimitedApprox}
\end{figure}

\subsection{Iterative Amplitude Sampling Reconstruction (IASR)}
The bandlimited interpolation algorithm (BIA) forms the basis for an iterative algorithm which we refer to as the Iterative Amplitude Sampling Reconstruction (IASR) algorithm as detailed in Algorithm \ref{alg:IASR} and illustrated in Fig.~\ref{fig:AmpSamp_Jane}. The time instants $\{t_n\}$ represent both samples of $f$ and the associated function $h$. Similarly to BIA, the sample values $\{h(n\Delta)=t_n-n\Delta/\alpha\}$ are the input to the algorithm assuming we also know the parameters $\alpha$ and $\Delta$ involved in the sampling process. Note that if the initialization satisfies $h_0\equiv0$ and $f_0\equiv0$, the first iteration corresponds precisely to BIA, i.e. $f_1(t)=\tilde{e}_1(t)=\tilde{f}(t)$ for all $t\in\mathbb{R}$. Thus, the emphasis is placed on the reconstruction of $h$ from its uniform amplitude samples and then the bandlimited constraint is imposed on the successive approximations to $f$ with the objective of iteratively reducing the error $||\tilde{e}_k(t)||_2$.

\begin{figure}[htb]
\begin{minipage}[b]{1.0\linewidth}
  \centering
  \centerline{\includegraphics[width=0.9\textwidth]{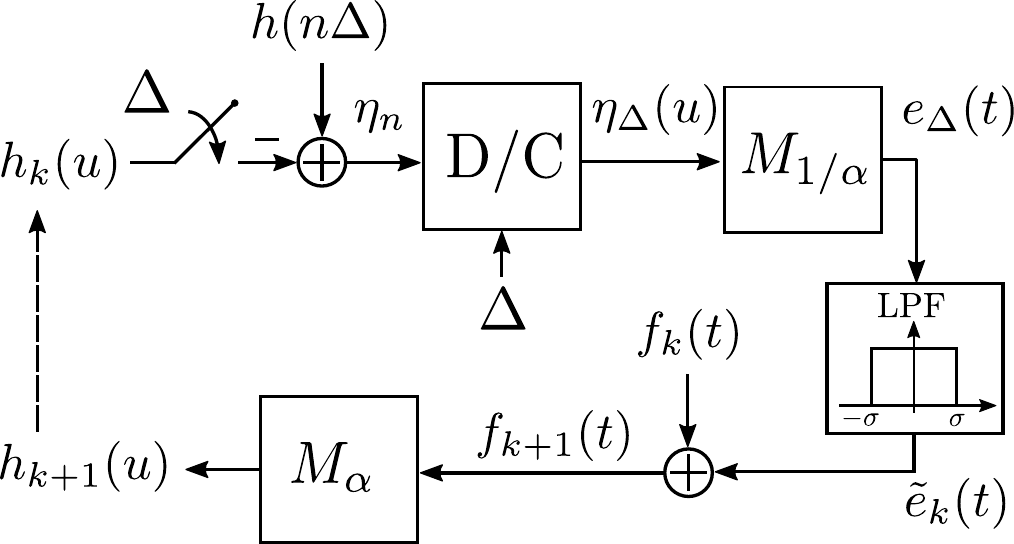}}
\end{minipage}
\caption{Block diagram representation of the iterative amplitude sampling reconstruction (IASR) algorithm.} 
\label{fig:AmpSamp_Jane}
\end{figure}

\begin{algorithm}
    \caption{IASR algorithm}
    \label{alg:IASR}
    \begin{algorithmic}[1] 
        \State \textbf{Input:} $\{t_n\}$, $\alpha$, $\Delta$, and $\sigma$
        \State $\{h(n\Delta)\}\gets \{t_n-n\Delta/\alpha\}$
        \State Initialize $h_0$ and $f_0$
        \Do
        	\State $\{\eta_n\}\gets\{h(n\Delta)-h_k(n\Delta)\}$
        	\State $\eta_{\Delta}(u)\gets\sum_{n\in\mathbb{Z}}\eta_n\mathrm{sinc}(u/\Delta-n)$
        	\State $e_{\Delta}\gets M_{1/\alpha} \eta_\Delta$
        	\State $\tilde{e}_k\gets \textrm{LPF}_\sigma(e_\Delta)$\Comment{$\textrm{LPF}_{\sigma}(\cdot)$ represents a lowpass filtering operation with cutoff frequency $\sigma$ rad/s.}
        	\State $f_{k+1}\gets f_k+\tilde{e}_k$
        	\State $h_{k+1}\gets M_\alpha f_{k+1}$
        \doWhile{$0\leq k<K$ \textbf{or} $||\tilde{e}_k||_2>\epsilon$}\Comment{$K$ and $\epsilon$ are parameters establishing the stopping criteria.}
        \State \textbf{return} $f_k$
    \end{algorithmic}
\end{algorithm}

\subsection{Simulation Results}
In all of the simulations in this section, the source signal $f$ is chosen as white noise bandlimited to $\sigma$ rad/s and bounded by $A>0$. The quantization step size $\Delta$ and the parameter $\alpha$ are chosen so that the sampling density is greater than or equal to the Landau rate \cite{Landau:1967aa}, which, in our case, is given by $\pi/\sigma$. We choose as a measure of approximation error the signal-to-error ratio (SER) given by
\begin{equation}
\mathrm{SER}=10\log_{10}\Big(\frac{||f||^2_2}{||f-f_k||^2_2}\Big)
\end{equation}
where $f_k$ is the $k$-th iteration.

Since amplitude sampling also implies nonuniform time sampling on the source signal $f$, we also directly apply a nonuniform reconstruction algorithm to recover $f$ in order to illustrate the particular factors influencing the performance of IASR. Specifically, we compare IASR to the Voronoi method developed in \cite[Theorem 8.13]{Feichtinger:1994aa} that presents the best tradeoff between convergence rate and approximation error among the frame-based methods described therein. We chose the Voronoi method since it has been shown to present a convergence rate approximately the same as reconstruction from local averages which has been used in other time encoding techniques for bandlimited signals \cite{Lazar:2003aa}. Based on the bounds in (\ref{eq:boundssamplingset}) for the time instants, it is straightforward to see that the sampling instants in an amplitude sampling setting satisfy the requirements of the Voronoi method for an appropriate choice of the parameters. In particular, it can be shown that it is sufficient that
\begin{equation}
\frac{\Delta}{|\alpha|-A\sigma}>\frac{\pi}{\sigma}.
\end{equation}
In initializing both algorithms, the $0$-th iteration in both IASR and the Voronoi method is assumed to be zero.

\begin{figure}[htb]
\begin{minipage}[b]{\linewidth}
  \centering
  \centerline{\includegraphics[width=0.9\linewidth]{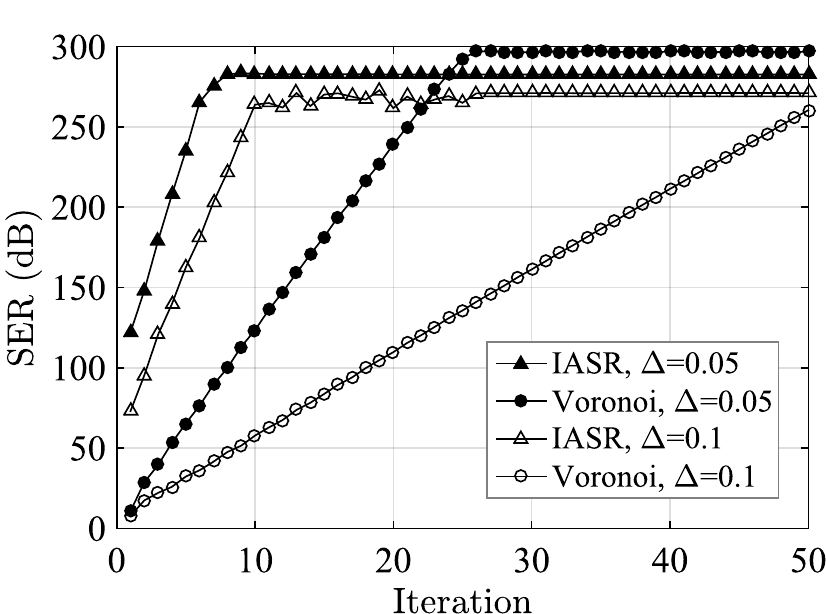}}
  \vspace{0.1cm}
  \centerline{(a)}\medskip
\end{minipage}
\begin{minipage}[b]{\linewidth}
  \centering
  \centerline{\includegraphics[width=0.9\linewidth]{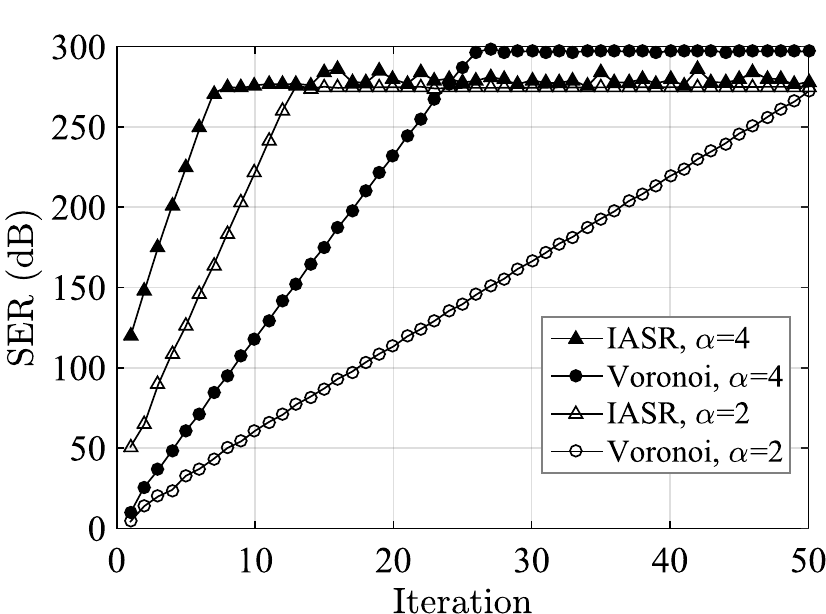}}
  \vspace{0.1cm}
\centerline{(b)}\medskip
\end{minipage}
\caption{Performance comparison between IASR, AWM, and BIA, for a broadband input signal bandlimited and bounded; (a) $\Delta$ is changed while $\alpha$ is fixed; (b)  $\alpha$ is changed while $\Delta$ is fixed.}
\label{fig:AlphaDelta_sims}
\end{figure}

In Fig.~\ref{fig:AlphaDelta_sims}, we have modified separately the parameters $\alpha$ and $\Delta$. In reducing the value of the quantization step size $\Delta$, the transformed function $\alpha t+f(t)$ will clearly cross more amplitude levels per unit of time. Similarly, when the slope of the ramp added to $f$ is increased in absolute value, it also causes an increase in the level-crossing density. Thus, both effects result in an increase of the sampling density, and as shown in the figure, the rate of convergence improves. However, it can be observed that the rate of convergence is faster in the IASR case. The first iteration in IASR achieves a better approximation than the Voronoi method although the rate of convergence appears to be highly insensitive to this change of parameters. On the other hand, the Voronoi method is significantly impacted by the change in sampling density. Moreover, it requires several iterations until it obtains the same approximation performance as the first iteration in IASR. As shown in Fig.~\ref{fig:BW_PabloMN} the same conclusions hold if we increase the oversampling ratio by considering signals with smaller bandwidths and, at the same time, keeping $\alpha$ and $\Delta$ fixed.

Thus far, we have focused on modifying the sampling density. Additionally, due to the structure of the sampling process in amplitude sampling, it is also possible to keep the sampling density fixed while changing both $\alpha$ and $\Delta$ accordingly. The rate of convergence in the Voronoi method is determined by the maximal separation between consecutive sampling instants. With constant sampling density, the performance of the Voronoi method does not change, as shown in Fig.~\ref{fig:FixedDiff_sims}. However, IASR presents an improvement in the rate of convergence. This is not surprising since the difference $|\alpha|-A\sigma$ has increased, which likely results in a better approximation of the sinc interpolation in IASR.

When the input signal is highly oversampled, we have empirically observed that the Voronoi method has a faster rate of convergence. Nevertheless, when the sampling instants become increasingly sparse approaching the Landau rate, IASR performs significantly better.

In summary, overall, IASR appears to have better performance than the Voronoi method in terms of speed of convergence when the sampling density approaches the Landau rate. Changes in the sampling density have a higher impact on the convergence in the Voronoi method than in IASR. Moreover, IASR performance can also be improved by increasing the difference $|\alpha|-A\sigma$ while keeping the sampling density invariant. In \cite{Lai:2016ab}, a scaling of the input signal also produces an increase in the speed of convergence. This performance improvement of IASR over the Voronoi method may be due to the characteristics of the sampling instants. Specifically the sampling instants in IASR inherently incorporate the amplitude sampling structure and therefore contain more information initially than more general nonuniform sampling would. In some sense, this may suggest that IASR is designed to more effectively exploit the structure of this particular sampling process which implicitly is signal dependent and consequently signal information is implicitly embedded in both the sampling times and the sample values.

\begin{figure}[htb]
\begin{minipage}[b]{0.9\linewidth}
  \centering
  \centerline{\includegraphics[width=\textwidth]{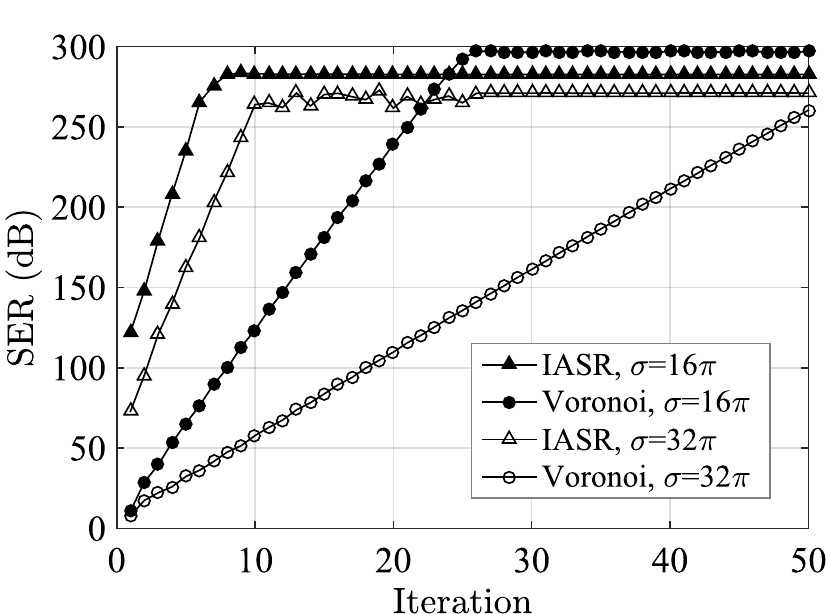}}
\end{minipage}
\caption{Performance comparison between IASR and the Voronoi method when the bandwidth $\sigma$ is changed, and $\alpha$ and $\Delta$ are fixed.}
\label{fig:BW_PabloMN}
\end{figure}

\begin{figure}[htb]
\begin{minipage}[b]{0.9\linewidth}
  \centering
  \centerline{\includegraphics[width=\textwidth]{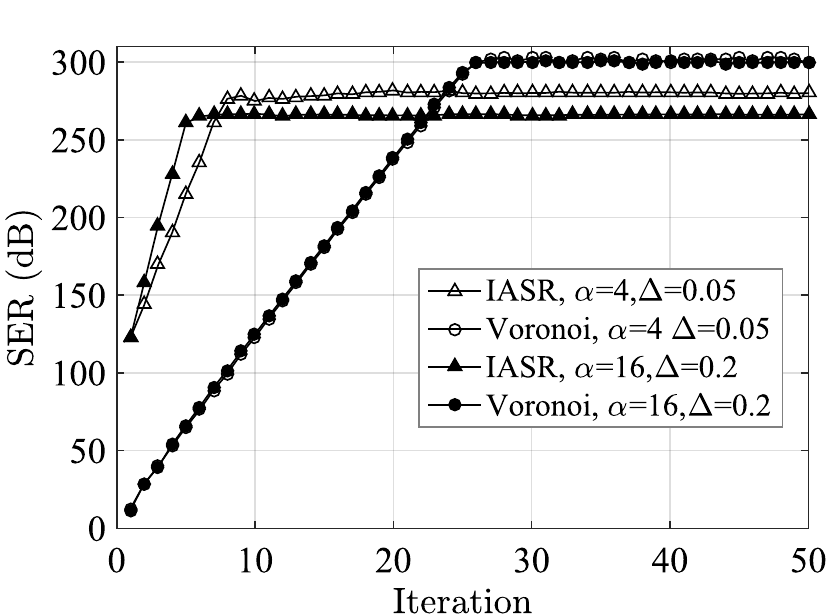}}
\end{minipage}
\caption{Performance comparison between IASR and the Voronoi method when the sampling density is fixed and $\alpha$ and $\Delta$ are changed.}
\label{fig:FixedDiff_sims}
\end{figure}

\subsection{Computational Complexity}
In the previous results, we have compared the performance of both algorithms in terms of convergence rate. Regarding operations per iteration, the IASR algorithm consists mainly of bandlimited interpolation, a lowpass filtering operation, and the two transformations $M_\alpha$ and $M_{1/\alpha}$. The bandlimited interpolation can be equivalently seen as a lowpass filtering operation. From a theoretical perspective, both transformations $M_\alpha$ and $M_{1/\alpha}$ only entail the addition of a ramp where the inverse can be interpreted through a relabeling of the axes $t$ and $u=g(t)$. However, in practice, we have considered in our simulations uniformly oversampled finite-length signals. Thus, in a practical setting, transformations of the form $M_{\alpha}$ take uniform samples to nonuniform samples. Due to the oversampling, we chose to perform linear interpolation to obtain an approximation to the corresponding uniform samples. 

The Voronoi method mainly consists of a zero-order hold and a lowpass filtering operation. Then, obtaining the function values at the nonuniform instants of time, i.e. $f_k(t_n)$, requires additional computation since they are not directly given by the first two operations. 

Fig.~\ref{fig:SERoverTime} shows a comparison of the reconstruction accuracy in terms of computation time for both methods. It is important to emphasize that these results depend on the software implementation and the hardware platform. In this case, we used the MATLAB environment. Among the several ways of simulating a lowpass filtering operation, we utilized the FFT and IFFT algorithms. From the simulations performed, we observe that the computation time per iteration is comparable. This suggests that the conclusions drawn above---i.e. when considering reconstruction accuracy versus number of iterations---can also apply here.

\begin{figure}[htb]
\begin{minipage}[b]{0.97\linewidth}
  \centering
  \centerline{\includegraphics[width=\textwidth]{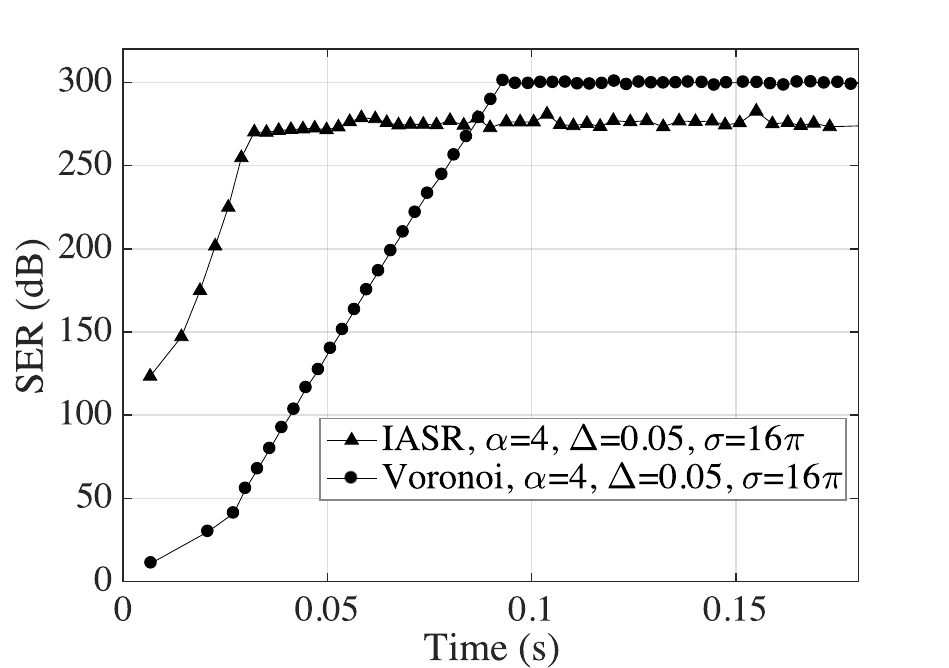}}
\end{minipage}
\caption{Performance comparison between IASR and the Voronoi method based on computation time for a specific hardware and software implementation. The solid markers represent the reconstruction accuracy of the corresponding iteration.}
\label{fig:SERoverTime}
\end{figure}

\section{Conclusion}
We presented the delta-ramp encoder, a form of time encoding where the time sequence generated is directly related to both uniform and nonuniform sampling of the corresponding associated signals. This sampling process can be analyzed theoretically as amplitude sampling which represents a signal with equally-spaced amplitude values and infinite-precision timing information by reversibly transforming the source signal. This transformation provides the perspective of viewing amplitude sampling as uniform sampling of an associated amplitude-time function or equivalently as nonuniform time sampling of the source signal. The properties of both functions are connected by a duality relationship.  Similarly, an iterative algorithm for recovery was proposed and evaluated that exploits the particular characteristics of the sampling instants in amplitude sampling. As opposed to more general nonuniform reconstruction algorithms, the convergence rate can be improved while maintaining the sampling density constant.




\appendices
\section{Proof of Theorem \ref{thm:IterationTheorem}}
\label{app:iterativeMalpha}
According to (\ref{eq:JanesEq_modified}), obtaining $\tilde{h}(u_0)$ for some fixed $u_0\in\mathbb{R}$ is equivalent to finding some $t^*\in\mathbb{R}$ such that $u_0=t^*+f(t^*)/\alpha$ since $\tilde{h}(u_0)=f(t^*)$. Therefore, we have to find the roots of 
\begin{equation}
t=u_0-\frac{1}{\alpha}f(t)\triangleq v_{u_0}(t)
\end{equation}
for $t\in\mathbb{R}$. It is easy to see that $v_{u_0}(t)$ is Lipschitz continuous for some constant $K<1$. Furthermore, there always exists some $\epsilon\geq A/\alpha$ such that $v_{u_o}:I\to I$ where $I=[u_0-\epsilon,u_0+\epsilon]$. Thus, the Banach fixed-point theorem \cite{Banach:1922aa} guarantees the uniqueness and existence of a solution. Moreover, it ensures convergence with the following bounds for the error
\begin{equation}
|t_{n+1}-t^*|\leq K|t_n-t^*|
\end{equation}
for $n\geq0$ where $t_0\in I$ and $t_{n+1}=v_{u_0}(t_n)$. Then, the iteration can be equivalently expressed in terms of the functional composition form as
\begin{equation}
\tilde{h}_{n+1}(u_0)=f(u_0-\frac{1}{\alpha}\tilde{h}_n(u_0)).
\end{equation}
Since $u_0$ was chosen arbitrarily, the same conclusions hold for any $u_0\in\mathbb{R}$.
\hfill$\square$

\section{Proof of Theorem \ref{thm:h_expdecay}}
\label{app:h_expdecay}
For ease of notation, we will substitute the real variable $t$ by $x$ throughout the proof of the theorem (i.e. we will refer to $f(x)$ instead of $f(t)$). Then, the complex variable $z\in\mathbb{C}$ is expressed as $z=x+iy$ for $x,y\in\mathbb{R}$. We will also use the complex variable $w=u+iv$ for $u,v\in\mathbb{R}$ when appropriate. Define the open disk in the complex plane centered at $z_o$ and of radius $r>0$ as 
\begin{equation}
D_r(z_o)=\{z\in\mathbb{C}:|z-z_o|<r\}
\end{equation}
and use $\overline{D}_r(z_o)$ for its closure.

Let us introduce some concepts that will be used throughout the proof. A function complex differentiable at every point in a region $\Omega\subseteq\mathbb{C}$ is said to be holomorphic in $\Omega$. If the function is holomorphic over the whole complex plane, it is referred to as entire. An entire function $f$ is of exponential type if there exists constants $M,\tau>0$ so that $|f(z)|\leq Me^{\tau|z|}$ for all $z\in\mathbb{C}$. If $\sigma=\inf \tau$ taken over all $\tau$ satisfying the latter inequality, it is said to be of exponential type $\sigma$.

We first introduce several results that will become useful in the proof of the theorem. 

\begin{lemma}
\label{lem:univalence}
Let $f$ be a holomorphic function in some region $\Omega\subseteq\mathbb{C}$ with power series $f(z)=\sum_{n=0}^{\infty}a_n(z-z_o)^n$ at $z_o\in\Omega$. Consider a disk of radius $R$ centered at $z_o$ such that $\Omega$ contains the disk and its closure. If $a_1\neq0$ and 
\begin{equation}
\label{eq:lemma1}
|a_1|>\sum_{n=2}^{\infty}|a_n|nR^{n-1}, 
\end{equation}
then $f$ is injective in any open disk of radius $r\leq R$.
\end{lemma}
\begin{proof}
Without loss of generality assume $z_o=0$, thus the power series expansion of $f$ around the origin is given by $f(z)=\sum_{n=0}^{\infty}a_nz^n$
for all $z\in\Omega$. Take $z_1, z_2\in\overline{D}_R(0)\subset\Omega$ such that $z_1\neq z_2$ and recall that for any $z,w\in\mathbb{C}$ the following identity holds
\begin{equation}
\begin{split}
(z^n-w^n)=(z-w)(z^{n-1}+&z^{n-2}w+\ldots\\
\ldots+&zw^{n-2}+w^{n-1}).
\end{split}
\end{equation}
We can write
\begin{eqnarray*}
\Big|\frac{f(z_2)-f(z_1)}{z_2-z_1}\Big| & = & \Big|a_1+\sum_{n=2}^{\infty}a_n\frac{z_2^n-z_1^n}{z_2-z_1}\Big| \\
& = & \Big|a_1+\sum_{n=2}^{\infty}a_n(z_2^{n-1}+z_2^{n-2}z_1+\ldots \\
&&\qquad\quad\ldots+z_2z_1^{n-2}+z_1^{n-1})\Big|\\
& \geq & |a_1|-\sum_{n=2}^{\infty}|a_n|nR^{n-1}.
\end{eqnarray*}
where the last inequality follows from the reverse triangle inequality and the fact that $|z_1|, |z_2|\leq R$. Thus, if $|a_1|-\sum_{n=2}^{\infty}|a_n|nR^{n-1}>0$, then $f(z_2)-f(z_1)\neq0$ and $f(z)$ is injective in $D_r(z_o)$ for any $r\leq R$.
\hfill$\square$\end{proof}

\begin{proposition}
\label{proposition:boundarytoboundary}
Let $f: U\to V $ be a bijective continuous function. Consider a set $\Omega$ such that its closure $\overline\Omega$ is strictly contained in $U$, then $f(\partial\Omega)=\partial f(\Omega)$.
\end{proposition}
\begin{proof}
Consider $x\in\partial\Omega$, which is clearly a limit point of $\Omega$. We know there exist a convergent sequence $x_n\to x$ where $x_n\neq x$ for $n\geq1$ and $x_n\in\Omega$. By continuity, $f(x_n)\to f(x)$ is a convergent sequence in $V$. As $f$ is a bijection from $U$ to $V$, $f(x_n)\neq f(x)$ for all $n\geq1$, thus $f(x)\in\overline{f(\Omega)}$. Moreover, $f(x)\neq f(x')$ for all $x'\in\Omega$, therefore $f(x)\in\partial f(\Omega)$ for all $x\in\partial\Omega$. It follows that $f(\partial\Omega)\subseteq\partial f(\Omega)$. 

Now, we claim that for every $y\in\partial\Omega$ there exists an $x\in\partial\Omega$ such that $y=f(x)$. Imagine this is not true and there exists an $x_o\in U\setminus\partial\Omega$ such that $y_o=f(x_o)$. From our previous discussion, it is clear that $x_o$ cannot be in $\Omega$, then imagine $x_o\in U\setminus\overline\Omega$. Since $f$ is continuous and bijective, we can choose a sufficiently small $\epsilon>0$ such that $f^{-1}(D_\epsilon(y_o))\subset D_\delta(x_o)$ and $D_\delta(x_o)\cap\overline\Omega=\emptyset$ for some $\delta>0$. However, as $y_o$ is a point in the boundary, it holds that $D_\epsilon(y_o)\cap f(\Omega)\neq\emptyset$. Thus, there exist an $x_1\in\Omega$ such that $f(x_1)=y'$ for some $y'\in D_\epsilon(y_o)$. At the same time, there also exists an $x_2\in f^{-1}(D_\epsilon(y_o))$ such that $f(x_2)=y'$, where $x_1\neq x_2$. This contradicts the bijectivity assumption, thus $f(\partial\Omega)\supseteq\partial f(\Omega)$ which together with $f(\partial\Omega)\subseteq\partial f(\Omega)$ gives $f(\partial\Omega)=\partial f(\Omega)$.
\hfill$\square$\end{proof}

\begin{proposition}
\label{proposition:fc_decrease}
Suppose $f(z)$ is an entire function of exponential type $\sigma$ such that $|f(x)|\leq A/(1+x^2)$ for all $x\in\mathbb{R}$. Then, the following bound holds for all $z\in\mathbb{C}$
\begin{equation}
\label{eq:proposition_fc_decrease}
|f(z)|\leq\frac{Ae^{\sigma|y|}}{1+x^2}.
\end{equation}
\end{proposition}
\begin{proof}
By assumption, $|f(z)|\leq Ae^{\sigma|z|}$ for all $z\in\mathbb{C}$. Construct the function $F(z)=(1/A)(1+x^2)e^{i\sigma z}f(z)$, then $F$ is bounded by 1 on the positive imaginary and positive real axis. If we consider the first quadrant $Q=\{z\in\mathbb{C}:x>0,y>0\}$, it is clear that there exists constants $C,c>0$ such that $|F(z)|\leq Ce^{c|z|}$ for $z\in Q$. We conclude by the Phragm\'en-Lindel\"of theorem \cite[Chapter 4, Theorem 3.4]{Stein:2003aa} that $|F(z)|\leq1$ for all $z$ in $Q$. This implies that $|f(z)|\leq Ae^{\sigma y}/(1+x^2)$ for $z\in Q$. Using the same argument, one can show that the same is true in the second quadrant. For the third and fourth quadrants we use instead the function $F(z)=(1/A)(1+x^2)e^{-i\sigma z}f(z)$, which shows that (\ref{eq:proposition_fc_decrease}) also holds for $y\leq0$.
\hfill$\square$\end{proof}

The function $f$ is of moderate decrease and bandlimited to $[-\sigma,\sigma]$ rad/s. By the Paley-Wiener theorem \cite[Chapter 4, Theorem 3.3]{Stein:2003aa}\cite[Theorem X]{Paley:1934aa}, $f$ is an entire function of exponential type $\sigma$. Then, using Bernstein's inequality \cite{Bernstein:1926aa}, $|f'(x)|\leq \sigma A$. Now, we can split the proof of the theorem in three steps.

\textit{Step 1.} We claim that the function $u=g(x)=\alpha x+f(x)$ admits a real analytic inverse function whenever $\alpha>\sigma A$. It is clear that $g(x)$ is analytic for all $x\in\mathbb{R}$ since it is the sum of two analytic functions on the whole real line. Moreover, $g(x)$ is a strictly increasing monotone function because $|f'(x)|\leq A\sigma$ and $\alpha>A\sigma$, which implies $g'(x)> 0$. The Real Analytic Inverse Function theorem \cite[Theorem 1.4.3]{Krantz:2002ab} guarantees that for a point $x_o$ where $g'(x_o)\neq0$, there exists a neighborhood $J_o$ of $x_o$ and a real analytic function $g^{-1}$ defined on an open interval $I_o$ containing $g(x_o)$ satisfying $(g^{-1}\circ g)(x)=x$ for $x\in J_o$ and $(g\circ g^{-1})(u)=u$ for $y\in I_o$. Since $g'(x)\neq0$ for all real $x$, it is always possible for any given $x_1\in\mathbb{R}$ to find an $x_2\notin J_1$ such that $J_1\cap J_2\neq\emptyset$. Thus, by analytic continuation, we conclude that $g^{-1}(u)$ is analytic on the whole real line.

\textit{Step 2.} We show that the function $g^{-1}(w)$ is analytic in a region containing the horizontal strip
\begin{equation}
\begin{split}
S_a&=\{w\in\mathbb{C}: |\mathrm{Im}(w)|<a,\ \\
&\qquad\mathrm{where}\ a=\frac{\alpha}{\sigma}\log\Big(\frac{\alpha}{\sigma}\Big)-\frac{\alpha-\sigma}{\sigma}\}.
\end{split}
\end{equation}
The function $g(z)=\alpha z+f(z)$ is an entire function of exponential type $\sigma$ and admits a power series expansion around $x\in\mathbb{R}$
\begin{equation}
g(z)=\alpha z+\sum_{n=0}^{\infty}\frac{f^{(n)}(x)}{n!}(z-x)^n
\end{equation}
for all $z\in\mathbb{C}$. By Bernstein's inequality \cite{Bernstein:1926aa}, the derivatives of $f$ are bounded on the real line by $|f^{(n)}(x)|\leq A\sigma^n$. We now look for a region where $g(z)$ is injective. Using Lemma \ref{lem:univalence}, $g(z)$ is injective in a disk of radius $R>0$ whenever
\begin{equation}
|\alpha+f'(x)|>\frac{A}{R}\sum_{n=2}^{\infty}n\frac{(\sigma R)^n}{n!}=A\sigma(e^{R\sigma}-1)
\end{equation}
or, equivalently
\begin{equation}
\label{eq:Rupperbound}
R<\frac{1}{\sigma}\log\Big(1+\frac{|\alpha+f'(x)|}{A\sigma}\Big).
\end{equation}
The right-hand side of this expression is lower bounded by $(1/\sigma)\log(1+(\alpha-A\sigma)/A\sigma)>0$, since $|f'(x)|\leq A\sigma<\alpha$ for all $x\in\mathbb{R}$. Thus, it is always possible to choose a disk of positive radius satisfying this lower bound such that $g(z)$ is injective.

Let us fix an $R$ satisfying this lower bound. Remember that holomorphic functions are open mappings, i.e. they map open sets to open sets. Thus, $g(z)$ maps an open disk of radius $R$ to the open set $g(D_R(x))$. By continuity, $g(D_R(x))$ is also connected since $D_R(x)$ is connected. Therefore, the mapping $g(z): D_R(x)\to g(D_R(x))$ represents a holomorphic bijection, thus its inverse is also holomorphic. Moreover, the inverse agrees with $g^{-1}(u)$ for real $u\in g(D_R(x))$. Thus, it represents the analytic continuation of $g^{-1}(u)$ on $u\in g(D_R(x))$. In fact, we can always choose a disk $D_R(x')$ such that $g(D_R(x))\cap g(D_R(x'))\neq\emptyset$, where the inverse functions defined on their respective images take the same value in the intersection for real $u$. Again, by analytic continuation, we can analytically extend $g^{-1}(u)$ to $g(D_R(x))\cup g(D_R(x'))$. Repeating this process for all real $x$, we obtain the analytic continuation of $g^{-1}(u)$ in the open set $\Omega=\cup_{x\in\mathbb{R}}g(D_R(x))$.

We want to find an $a>0$ such that $S_a\subseteq\Omega$. Using Lemma \ref{proposition:boundarytoboundary}, the boundary of the disk $\partial D_R(x)$ is mapped bijectively to $\partial g(D_R(x))$. Therefore, the largest radius $\rho$ for a disk centered at $g(x)$ such that $D_{\rho}(g(x))\subseteq g(D_R(x))$ for all $x\in\mathbb{R}$ is given by
\begin{equation}
\begin{split}
\rho&=\inf_{x\in\mathbb{R}}\sup_{|z-x|=R}\{|g(z)-g(x)|:\\
&\qquad D_{|g(z)-g(x)|}(g(x))\subseteq g(D_R(x))\}.
\end{split}
\end{equation}
We can use the power series expansion of $g$ around $x$ to find a lower bound for $\rho$ in the following manner
\begin{eqnarray*}
|g(z)-g(x)| & = & |\alpha (z-x)+\sum_{n=1}^{\infty}a_n(z-x)^n|\\
 & \geq & \alpha R-\sum_{n=1}^{\infty}\frac{A\sigma^n}{n!}R^n=\alpha R-A(e^{R\sigma}-1).
\end{eqnarray*}
for $|z-x|=R$. The right-hand side of the last expression represents a strictly concave function of $R$, thus the maximum is achieved for
\begin{equation}
\label{eq:Rmax}
R=\frac{1}{\sigma}\log\Big(\frac{\alpha}{A\sigma}\Big)>0
\end{equation}
which is positive as $\alpha>A\sigma$ and satisfies the upper bound in (\ref{eq:Rupperbound}). Setting the value of $R$ as in (\ref{eq:Rmax}), we can write
\begin{equation}
\label{eq:abound}
|g(z)-g(x)|\geq\rho\geq\frac{\alpha}{\sigma}\log\Big(\frac{\alpha}{A\sigma}\Big)-\Big(\frac{\alpha-A\sigma}{\sigma}\Big)
\end{equation}
for all $x\in\mathbb{R}$ and $|z-x|=R$. This implies that $S_a\subseteq\Omega$ for any $a$ such that
\begin{equation}
\label{eq:abound_2}
a\leq\frac{\alpha}{\sigma}\log\Big(\frac{\alpha}{A\sigma}\Big)-\Big(\frac{\alpha-A\sigma}{\sigma}\Big).
\end{equation}

\textit{Step 3.}
We show that $h(w)$ is of moderate decay on each horizontal line $|\mathrm{Im}(w)|<a$, uniformly in $|y|<a$. First, we note that since $f$ is an entire function of exponential type $\sigma$ and is of moderate decrease along the real line, by Proposition \ref{proposition:fc_decrease}
\begin{equation}
|f(z)|\leq\frac{Ae^{\sigma|y|}}{1+x^2}
\end{equation}
for all $z\in\mathbb{C}$. Let us now fix an $R$ satisfying (\ref{eq:Rmax}), then we have a bijection from $D_R(x')$ to $g(D_R(x'))$ for some $x'\in\mathbb{R}$. Therefore, $z=g^{-1}(w)$, where $w\in g(D_R(x))$ and $z\in D_R(x')$. Since $|y|<R$ for $z\in D_R(x')$, we also have $|f(z)|=|\alpha z-g(z)|\leq Ae^{\sigma R}/(1+x^2)$, or equivalently 
\begin{equation}
\label{eq:gc_bounds}
|w-\alpha g^{-1}(w)|\leq\frac{Ae^{\sigma R}}{1+(g^{-1}(u))^2}
\end{equation}
whenever $w\in g(D_R(x))$ and $z\in D_R(x')$. Using the reverse triangle inequality in the previous expression for real $w$, we can also obtain
\begin{equation}
\label{eq:lowg_bounds}
|g^{-1}(u)|\geq\frac{|u|}{\alpha}-\frac{Ae^{\sigma R}}{\alpha}.
\end{equation}
which is true for all $u\in\mathbb{R}$ since $x=g^{-1}(u)$ holds for all real $x$ and $u$ as shown in the first step of the proof. Define the function for all real $u$
\begin{equation}
\psi(u) =
  \begin{cases} 
      \hfill |u|/\alpha-Ae^{\sigma R}/\alpha\hfill&\text{ if $|u|/\alpha>Ae^{\sigma R}/\alpha$}\\
      \hfill0\hfill&\text{ otherwise}\\
  \end{cases}
\end{equation}
which clearly satisfies $|g^{-1}(u)|\geq|\psi(u)|$. Make $\beta=1/\alpha$ and multiply both sides of (\ref{eq:gc_bounds}) by $1/\alpha$ to see that $|h(w)|=|w/\alpha-g^{-1}(w)|$. Combining these expressions, we can then write for some $A'>0$ and $w\in g(D_R(x'))$
\begin{equation}
|h(w)|\leq\frac{Ae^{\sigma R}/\alpha}{1+g^{-1}(u)^2}\leq\frac{Ae^{\sigma R}/\alpha}{1+\psi(u)^2} \leq \frac{A'}{1+u^2}
\end{equation}
As our choice of $x'$ was arbitrary, this is true for any $x'\in\mathbb{R}$ and $|h(u+iv)|$ is of moderate decrease along horizontal lines. 

Therefore, the function $h(w)$ is analytic on the strip $S_a$ and it is of moderate decrease on each horizontal line $|\mathrm{Im}(w)|=v$, uniformly in $|v|<a$, as long as $\beta=1/\alpha$. By \cite[Chapter 4, Theorem 2.1]{Stein:2003aa}, we conclude that there exists a constant $C>0$ such that $|\hat{h}(\xi)|\leq Ce^{-2\pi b\xi}$ for any $0\leq b<a$.
\hfill$\square$

\section{Proof of Theorem \ref{thm:h_nonBL}}
\label{app:h_nonBL}
Construct the function $g(z)=\alpha z+f(z)$ where $f$ is not constant. By Picard's little theorem \cite[16.22]{Rudin:1986aa}, there exists at most one value $\alpha>A\sigma$ that $f'(z)$ does not take. For the rest of them, there always exists a $z_o\in\mathbb{C}$ such that $g'(z_o)=\alpha+f'(z_o)=0$. Then, it is possible to write
\begin{equation}
g(z)-a_0=(z-z_o)^2[a_2+a_3(z-z_o)+a_4(z-z_o)^2+\ldots].
\end{equation}
Therefore, the function $g(z)-a_o$ has a zero of order $\geq2$ at $z_o$. By the Local Mapping Theorem \cite[Chapter 3, Theorem 11]{Ahlfors:1966aa}, $g$ is $n$-to-1 near $z_o$ for $n\geq2$. Using the argument of analytic continuation of local biholomorphisms in the proof of Theorem \ref{thm:h_expdecay}, we conclude that the analytic extensions of $h$ around $g(z_o)$ are multivalued. This excludes the possibility of $h$ being entire, thus, the restriction of $h$ to the real line cannot be bandlimited.

If $f(z)=C$ for some $C>0$, then $h(u)=-C/\alpha$, which is bandlimited in the distributional sense.
\hfill$\square$

\section{Time-Domain Decay Properties}
\label{app:BoundednessSlopes}
\begin{proposition}
\label{prop:boundedness}
Under the conditions of Theorem \ref{thm:h_expdecay}, construct $h(u)=g^{-1}(u)-\beta u$ for $u\in\mathbb{R}$ and some $\beta\in\mathbb{R}$. Then, the function $h$ is of moderate decrease if and only if $\beta=1/\alpha$.
\end{proposition}
\begin{proof}
The backward direction, i.e. assuming $\beta=1/\alpha$, has been proved in Theorem \ref{thm:h_expdecay}. For the forward direction, assume on the contrary that $\beta\neq1/\alpha$. From (\ref{eq:gc_bounds}), we have the following bound
\begin{equation}
\Big||g^{-1}(u)|-|u|/\alpha\Big|\leq A/\alpha.
\end{equation}
Therefore, using the reverse triangle inequality and the previous expression, we arrive at the following for large enough $u$
\begin{equation}
\begin{split}
|h(u)|=|g^{-1}(u)-\beta u|&\geq\Big||g^{-1}(u)|-\beta|u|\Big|\\
&\geq\Big||u|\Big|\beta-\frac{1}{\alpha}\Big|-\frac{A}{\alpha}\Big|.
\end{split}
\end{equation}
Clearly, the right-hand side grows linearly without bound. Thus, $|h(u)|$ is unbounded for $\beta\neq1/\alpha$ and cannot be of moderate decrease.
\hfill$\square$\end{proof}

\begin{proposition}
\label{prop:Lpnorms}
Under the conditions of Theorem \ref{thm:h_expdecay}, the following relationship holds
\begin{equation}
||h||_p=\frac{1}{\alpha^{1-\frac{1}{p}}}||f||_p
\end{equation}
whenever $1\leq p<\infty$.
\end{proposition}
\begin{proof}
Using (\ref{eq:JanesEq}) and the change of variables $u=f(t)+\alpha t$, it is clear that for any even number $p$
\begin{equation}
\begin{split}
||h||_p^p&=\int_{\mathbb{R}}h(f(t)+\alpha t)^p(f'(t)+\alpha)\mathrm{d}t\\
&=\frac{1}{\alpha^p}\int_{\mathbb{R}}f(t)^pf'(t)\mathrm{d}t+\frac{1}{\alpha^{p-1}}||f||_p^p.
\end{split}
\end{equation}
We now show that the first term on the right-hand side vanishes. First note that by the fundamental theorem of calculus, $\int_{\mathbb{R}}f'(t)\mathrm{d}t=0$ for any continuous function satisfying $\lim_{|t|\to\infty}f(t)=0$. Thus, we arrive at the following
\begin{equation}
\int_{\mathbb{R}}f(t)^pf'(t)\mathrm{d}t=\frac{1}{p+1}\int_{\mathbb{R}}(f(t)^{p+1})'\mathrm{d}t=0
\end{equation}
since $f(t)^{p+1}$ is also of moderate decay. Therefore, $||h||_p=||f||_p/\alpha^{(p-1)/p}$.

Now, for any number $p\geq1$, we can choose a collection of intervals such that $f(t)\geq0$ for $t\in(a_n,a_{n+1})$ and $f(t)\leq0$ for $t\in(b_m,b_{m+1})$ for all $m,n\in\mathbb{Z}$. Thus, we can write
\begin{equation}
\label{eq:splitNorm}
\begin{split}
\int_{\mathbb{R}}|f(t)|^pf'(t)\mathrm{d}t&=\sum_{n\in\mathbb{Z}}\int_{a_n}^{a_{n+1}}f(t)^pf'(t)\mathrm{d}t\\
&-\sum_{m\in\mathbb{Z}}\int_{b_m}^{b_{m+1}}f(t)^pf'(t)\mathrm{d}t.
\end{split}
\end{equation}
For each $n\in\mathbb{Z}$, we then have that
\begin{equation}
\int_{a_n}^{a_{n+1}}f(t)^pf'(t)\mathrm{d}t=\frac{f(t)^{p+1}}{p+1}\Big|_{a_n}^{a_{n+1}}=0
\end{equation}
since $f(a_n)^{p+1}=0$ for all $n\in\mathbb{Z}$. The same is true for the second term of the right-hand side of (\ref{eq:splitNorm}). Note that we allow to have intervals of the form $(c,\infty)$ or $(-\infty,d)$ for any $c,d\in\mathbb{R}$ and the same holds true since $\lim_{|t|\to\infty}f(t)^{p+1}=0$ for $p\geq1$. Therefore,
\begin{equation}
\int_{\mathbb{R}}|f(t)|^pf'(t)\mathrm{d}t=0
\end{equation}
and $||h||_p=||f||_p/\alpha^{(p-1)/p}$ for any number $p\geq1$.
\hfill$\square$\end{proof}

\begin{proposition}
\label{prop:L1norms}
Under the conditions of Theorem \ref{thm:h_expdecay}, the following relationship holds
\begin{equation}
||h_1-h_2||_1=||f_1-f_2||_1
\end{equation}
where $h_1=M_\alpha f_1$ and $h_2=M_\alpha f_2$.
\end{proposition}
\begin{proof}
Note that $|h_1(u)-h_2(u)|=|g^{-1}_1(u)-g^{-1}_2(u)|$ and $|f_1(t)-f_2(t)|=|g_1(t)-g_2(t)|$, thus by Fubini's theorem \cite[Theorem 4.1.6]{Stroock:2011aa} we arrive at the following
\begin{eqnarray*}
||h_1-h_2||_1&=&\int_{\mathbb{R}}|g^{-1}_1(u)-g^{-1}_2(u)|\mathrm{d}u=\int_{\mathbb{R}^2}\mathbf{1}_{\Gamma}\mathrm{d}t\mathrm{d}u\\
||f_1-f_2||_1&=&\int_{\mathbb{R}}|g_1(t)-g_2(t)|\mathrm{d}t=\int_{\mathbb{R}^2}\mathbf{1}_{\Lambda}\mathrm{d}t\mathrm{d}u
\end{eqnarray*}
where
\begin{eqnarray*}
\Gamma&=&\{(t,u)\in\mathbb{R}^2: t\in\mathbb{R}, \\
&&\min\{g_1(t),g_2(t)\}\leq u\leq\max\{g_1(t),g_2(t)\}\}\\
\Lambda&=&\{(t,u)\in\mathbb{R}^2: u\in\mathbb{R}, \\
&&\min\{g^{-1}_1(u),g^{-1}_2(u)\}\leq t\leq\max\{g^{-1}_1(u),g^{-1}_2(u)\}\}.
\end{eqnarray*}
For an arbitrary $(t_o,u_o)\in\Gamma$, $t_o\leq\max\{g^{-1}_1(u_o),g^{-1}_2(u_o)\}$ and $t_o\geq\min\{g^{-1}_1(u_o),g^{-1}_2(u_o)\}$ since $g_1$ and $g_2$ are strictly increasing. This implies that $\Gamma\supseteq\Lambda$. Using the same reasoning, we obtain $\Gamma\subseteq\Lambda$ which implies that $\Gamma=\Lambda$. Thus, the integrals are the same.
\hfill$\square$\end{proof}

\section{}
\label{app:infinityerror}
\begin{proposition}
\label{prop:BLerror}
If $h_{\Delta }$ is the bandlimited approximation to $h$, then there exists a $C'>0$ such that 
\begin{equation}
||h-h_{\Delta}||_{\infty}\leq \frac{C'}{a}e^{-\pi\frac{b}{\Delta}}
\end{equation}
for any $0\leq b<a$ where $a=\frac{\alpha}{\sigma}\log(\frac{\alpha}{A\sigma})-\frac{\alpha-A\sigma}{\sigma}$.
\end{proposition}
\begin{proof}
Since $h(u)$ is continuous and of moderate decrease, $\hat{h}(\xi)$ is also continuous. Moreover, we know from Theorem \ref{thm:h_expdecay} that $|\hat{h}(\xi)|\leq Ce^{-b2\pi|\xi|}$ for some $C>0$ and $0\leq b<a$ where $a$ is defined in (\ref{eq:h_expdecay}). Clearly, $\hat{h}$ is Lebesgue measurable and absolutely integrable. Thus, by the generalized form of Weiss's theorem \cite{Brown:1967aa} we can bound the approximation error as
\begin{equation}
\begin{split}
|h(u)-h_{\Delta}(u)|&\leq 2\int_{|\xi|>1/2\Delta}|\hat{h}(\xi)|\mathrm{d}\xi\\
&\leq 4C\int_{\xi>1/2\Delta}e^{-2\pi\xi b}\mathrm{d}\xi=\frac{4C}{2\pi b}e^{-\pi\frac{b}{\Delta}}
\end{split}
\end{equation}
for all $u\in\mathbb{R}$.
\hfill$\square$\end{proof}

\ifCLASSOPTIONcaptionsoff
  \newpage
\fi



%

\bibliographystyle{IEEEtran}
\bibliography{/Users/pmnuevo/Documents/BibDeskLibrary/BibDeskLibrary}

\begin{thebibliography}{10}
\providecommand{\url}[1]{#1}
\csname url@samestyle\endcsname
\providecommand{\newblock}{\relax}
\providecommand{\bibinfo}[2]{#2}
\providecommand{\BIBentrySTDinterwordspacing}{\spaceskip=0pt\relax}
\providecommand{\BIBentryALTinterwordstretchfactor}{4}
\providecommand{\BIBentryALTinterwordspacing}{\spaceskip=\fontdimen2\font plus
\BIBentryALTinterwordstretchfactor\fontdimen3\font minus
  \fontdimen4\font\relax}
\providecommand{\BIBforeignlanguage}[2]{{%
\expandafter\ifx\csname l@#1\endcsname\relax
\typeout{** WARNING: IEEEtran.bst: No hyphenation pattern has been}%
\typeout{** loaded for the language `#1'. Using the pattern for}%
\typeout{** the default language instead.}%
\else
\language=\csname l@#1\endcsname
\fi
#2}}
\providecommand{\BIBdecl}{\relax}
\BIBdecl

\bibitem{Whittaker:1915aa}
E.~T. Whittaker, ``{XVIII}.---{O}n the functions which are represented by the
  expansions of the interpolation-theory,'' \emph{Proceedings of the Royal
  Society of Edinburgh}, vol.~35, pp. 181--194, 1915.

\bibitem{Kotelnikov:1933aa}
V.~A. Kotelnikov, ``On the carrying capacity of the ether and wire in
  telecommunications,'' in \emph{Material for the First All-Union Conference on
  Questions of Communication, Izd. Red. Upr. Svyazi RKKA, Moscow}, 1933.

\bibitem{Shannon:1949aa}
C.~E. Shannon, ``Communication in the presence of noise,'' \emph{Proceedings of
  the IRE}, vol.~37, no.~1, pp. 10--21, 1949.

\bibitem{Bond:1958aa}
F.~Bond and C.~Cahn, ``On sampling the zeros of bandwidth limited signals,''
  \emph{IRE Transactions on Information Theory}, vol.~4, no.~3, pp. 110--113,
  1958.

\bibitem{Logan:1977aa}
B.~F. Logan, ``Information in the zero crossings of bandpass signals,''
  \emph{Bell System Technical Journal}, vol.~56, no.~4, pp. 487--510, 1977.

\bibitem{Roweis:1998ac}
S.~Roweis, S.~Mahajan, and J.~Hopfield, ``Signal reconstruction from
  zero-crossings,'' August 1998, draft (accessed May 2016)
  https://www.cs.nyu.edu/~roweis/papers/logan.ps.

\bibitem{Haavik:1966aa}
S.~J. Haavik, ``The conversion of the zeros of noise,'' Master's thesis,
  University of Rochester, Rochester, NY, 1966.

\bibitem{Bar-David:1974aa}
I.~Bar-David, ``An implicit sampling theorem for bounded bandlimited
  functions,'' \emph{Information and Control}, vol.~24, no.~1, pp. 36--44,
  1974.

\bibitem{Kumaresan:2000aa}
R.~Kumaresan and Y.~Wang, ``A new real-zero conversion algorithm,'' in
  \emph{Proceedings 2000 IEEE International Conference on Acoustics, Speech,
  and Signal Processing.}, vol.~1, 2000, pp. 313--316 vol.1.

\bibitem{Duffin:1938aa}
R.~J. Duffin and A.~C. Schaeffer, ``Some properties of functions of exponential
  type,'' \emph{Bulletin of the American Mathematical Society}, vol.~44, no.~4,
  pp. 236--240, 1938.

\bibitem{Stein:2003aa}
E.~M. Stein and R.~Shakarchi, \emph{Complex analysis}, ser. Princeton Lectures
  in Analysis, II.\hskip 1em plus 0.5em minus 0.4em\relax Princeton University
  Press, Princeton, NJ, 2003.

\bibitem{Selva:2012aa}
J.~Selva, ``Efficient sampling of band-limited signals from sine wave
  crossings,'' \emph{{IEEE} Trans. Signal Processing}, vol.~60, no.~1, pp.
  503--508, 2012.

\bibitem{Kay:1986aa}
S.~Kay and R.~Sudhaker, ``A zero crossing-based spectrum analyzer,'' \emph{IEEE
  Transactions on Acoustics, Speech and Signal Processing}, vol.~34, no.~1, pp.
  96--104, 1986.

\bibitem{Sreenivas:1992aa}
T.~V. Sreenivas and R.~J. Niederjohn, ``Zero-crossing based spectral analysis
  and svd spectral analysis for formant frequency estimation in noise,''
  \emph{IEEE Transactions on Signal Processing}, vol.~40, no.~2, pp. 282--293,
  1992.

\bibitem{Kumaresan:2010aa}
R.~Kumaresan and N.~Panchal, ``Encoding bandpass signals using zero/level
  crossings: a model-based approach,'' \emph{Audio, Speech, and Language
  Processing, IEEE Transactions on}, vol.~18, no.~1, pp. 17--33, 2010.

\bibitem{Mallat:1991aa}
S.~Mallat, ``Zero-crossings of a wavelet transform,'' \emph{IEEE Transactions
  on Information Theory}, vol.~37, no.~4, pp. 1019--1033, 1991.

\bibitem{Mark:1981aa}
J.~W. Mark and T.~D. Todd, ``A nonuniform sampling approach to data
  compression,'' \emph{{IEEE} Trans. Commun.}, vol.~29, no.~1, pp. 24--32,
  1981.

\bibitem{Tsividis:2003aa}
Y.~Tsividis, ``Continuous-time digital signal processing,'' \emph{Electron.
  Lett.}, vol.~39, no.~21, pp. 1551--1552, 2003.

\bibitem{Inose:1966aa}
H.~Inose, T.~Aoki, and K.~Watanabe, ``Asynchronous delta-modulation system,''
  \emph{Electron. Lett.}, vol.~2, no.~3, pp. 95--96, 1966.

\bibitem{Lazar:2003aa}
A.~A. Lazar and L.~T. T{\'o}th, ``Time encoding and perfect recovery of
  bandlimited signals,'' in \emph{Acoustics, Speech, and Signal Processing,
  2003. Proceedings.(ICASSP'03). 2003 IEEE International Conference on},
  vol.~6.\hskip 1em plus 0.5em minus 0.4em\relax IEEE, 2003, pp. VI--709.

\bibitem{Feichtinger:1994aa}
H.~G. Feichtinger and K.~Gr{\"o}chenig, ``Theory and practice of irregular
  sampling,'' \emph{Wavelets: mathematics and applications}, pp. 305--363,
  1994.

\bibitem{McCormick:2012aa}
M.~McCormick, ``Digital pulse processing,'' Master's thesis, Massachusetts
  Institute of Technology, September 2012.

\bibitem{Martinez-Nuevo:2016aa}
P.~Mart\'inez-Nuevo, ``Amplitude sampling for signal representation,'' Ph.D.
  dissertation, Massachusetts Institute of Technology, 2016.

\bibitem{Lai:2016ab}
H.~Lai, ``Reconstruction methods for level-crossing sampling,'' Master's
  Thesis, EECS, MIT, June 2016.

\bibitem{Landau:1967aa}
H.~J. Landau, ``Necessary density conditions for sampling and interpolation of
  certain entire functions,'' \emph{Acta Mathematica}, vol. 117, no.~1, pp.
  37--52, July 1967.

\bibitem{Bernstein:1926aa}
S.~N. Bernstein, \emph{Le{\c{c}}ons sur les propri{\'e}t{\'e}s extr{\'e}males
  et la meilleure approximation des fonctions analytiques d'une variable
  r{\'e}elle}.\hskip 1em plus 0.5em minus 0.4em\relax Paris, 1926.

\bibitem{Azizi:1999aa}
S.~Azizi, D.~Cochran, and J.~McDonald, ``On the preservation of bandlimitedness
  under non-affine time warping,'' in \emph{Proc. Int. Workshop on Sampling
  Theory and Applications}, 1999.

\bibitem{Banach:1922aa}
S.~Banach, ``Sur les op{\'e}rations dans les ensembles abstraits et leur
  application aux {\'e}quations int{\'e}grales,'' \emph{Fund. Math}, vol.~3,
  no.~1, pp. 133--181, 1922.

\bibitem{Paley:1934aa}
R.~E. A.~C. Paley and N.~Wiener, \emph{Fourier transforms in the complex
  domain}.\hskip 1em plus 0.5em minus 0.4em\relax American Mathematical Soc.,
  1934, vol.~19.

\bibitem{Krantz:2002ab}
S.~G. Krantz and H.~R. Parks, \emph{A primer of real analytic functions}.\hskip
  1em plus 0.5em minus 0.4em\relax Springer Science \& Business Media, 2002.

\bibitem{Rudin:1986aa}
W.~Rudin, \emph{Real and complex analysis}, 3rd~ed.\hskip 1em plus 0.5em minus
  0.4em\relax McGraw-Hill Education, 1986.

\bibitem{Ahlfors:1966aa}
L.~V. Ahlfors, \emph{Complex analysis: an introduction to the theory of
  analytic functions of one complex variable}, 2nd~ed.\hskip 1em plus 0.5em
  minus 0.4em\relax McGraw-Hill Book Company, 1966.

\bibitem{Stroock:2011aa}
D.~W. Stroock, \emph{Essentials of integration theory for analysis}.\hskip 1em
  plus 0.5em minus 0.4em\relax Springer Science \& Business Media, 2011, vol.
  262.

\bibitem{Brown:1967aa}
J.~L. Brown, Jr., ``On the error in reconstructing a non-bandlimited function
  by means of the bandpass sampling theorem,'' \emph{J. Math. Anal. Appl.},
  vol.~18, pp. 75--84, 1967.

\end{thebibliography}
\end{document}